\documentclass[a4paper,UKenglish,cleveref, autoref]{lipics-v2019}

\title{Monadic Second-Order Logic with Path-Measure Quantifier is
  Undecidable\footnote{We became aware that this result was already proved by Mikołaj Bojańczyk, Edon Kelmendi and Michał
Skrzypczak through a work that was conducted independently of
ours~\cite{bojanczyk2019undecidable}. Our proof technique relies on probabilistic automata, which differs from the approach of Bojańczyk et
al.}
} 

\author{Rapha\"el Berthon}{Universit\'e libre de Bruxelles, Belgium\ }{}{}{}

\author{Emmanuel Filiot}{Universit\'e libre de Bruxelles, Belgium\ }{}{}{}

\author{Shibashis Guha}{Universit\'e libre de Bruxelles, Belgium\ }{}{}{}

\author{Bastien Maubert}{Universit\`a degli Studi di Napoli Federico II, Italy\ }{}{}{}

\author{Aniello Murano}{Universit\`a degli Studi di Napoli Federico
  II, Italy\ }{}{}{}

\author{Laureline Pinault}{ENS Lyon, France\ }{}{}{}

\author{Jean-Fran\c{c}ois Raskin}{Universit\'e libre de Bruxelles, Belgium\ }{}{}{}

\author{Sasha Rubin}{Universit\`a degli Studi di Napoli Federico II, Italy\ }{}{}{}

\authorrunning{R. Berthon et al.}

\Copyright{}



\category{}

\relatedversion{}

\supplement{}

\hideLIPIcs  

\EventEditors{}
\EventNoEds{2}
\EventLongTitle{}
\EventShortTitle{}
\EventAcronym{CVIT}
\EventYear{2016}
\EventDate{December 24--27, 2016}
\EventLocation{Little Whinging, United Kingdom}
\EventLogo{}
\SeriesVolume{42}
\ArticleNo{23}

\usepackage{amsmath, amssymb, amsfonts, amsthm}
\usepackage{pgf}
\usepackage{tikz}
\usetikzlibrary{arrows,automata, positioning}
\usepackage{bwc}
\usepackage{xspace}
\usepackage{todonotes}

\nolinenumbers

\usepackage{macros,colored-comments}

\begin{document}

\maketitle

\begin{abstract}
We consider an extension of monadic second-order logic, interpreted over the infinite binary tree, 
by the qualitative path-measure quantifier. 
This quantifier says that the set of infinite paths in the tree satisfying a 
formula has Lebesgue-measure one. We prove that this logic is undecidable. 
To do this we prove that the emptiness problem 
of qualitative universal parity tree automata is undecidable. Qualitative means 
that a run of a tree automaton is accepting if the set of paths in the run that 
satisfy the acceptance condition has Lebesgue-measure one.  
\end{abstract}

\section{Introduction}
Monadic Second-Order logic (MSO) is an extension of first-order logic
with quantification on sets. The fundamental result about MSO is that
its theory on 
the infinite binary tree is decidable~\cite{Rabin:TAMS69}. There are a 
number of ways to extend this result: to structures generated by certain 
operations (see the survey~\cite{Thomas:MFCS03}), by certain additional unary 
predicates~\cite{Fratani:TCS12}, and by certain generalised 
quantifiers~\cite{Rabinovich:FI10,michalewski2018monadic}. In this note we 
consider \MSOzeropath, the extension of \MSO by the measure-theoretic quantifier $\forallonep 
X$, introduced in~\cite{michalewski2016measure}. Here, 
$\forallonep X. \varphi$ 
states that the set 
of paths of the tree that satisfy $\varphi$ has Lebesgue-measure equal to $1$. This 
means, intuitively, that a random path almost surely satisfies $\varphi$, 
where a random path is generated by repeatedly flipping a fair coin to 
decide to go to left or right. A weak fragment of this logic is known
to be
decidable~\cite{bojanczyk2016thin,bojanczyk_et_al:LIPIcs:2017:7474},
and a more general one (in which the measure quantifier ranges over all sets
instead of just paths) was proved undecidable
in~\cite{michalewski2018monadic}. 

We prove that \MSOzeropath is undecidable by encoding the emptiness problem of qualitative universal parity 
tree automata (the encoding is direct). Such an automaton accepts a tree $t$ if every run $\rho$ on $t$ has the property that the 
Lebesgue-measure of the set of branches of $\rho$ satisfying the parity condition is equal to $1$. Thus, 
the main technical contribution of this note is that this emptiness problem is undecidable (Theorem~\ref{theo-stage}). 


\section{Preliminaries}
\label{sec-prelim}

Given a finite non-empty set $\Sigma$, called an alphabet, we write $\Sigma^*$ and $\Sigma^\omega$
for the set of finite and infinite words over $\Sigma$,
respectively. For a finite word $w=a_0\dots a_{n-1}$ we write $|w|=n$ for
its \emph{length}, and if $w$ is an infinite word we let
$|w|=\omega$. For a word $w$ and index $i < |w|$, we let $w_i$ be
the letter at position $i$ in $w$.
For a finite word $w\in\Sigma^*$, the set $\cone(w)=w\cdot \Sigma^\omega$ 
of infinite words is called the \emph{cone} of $w$. 
For a function $f:A\rightarrow B$,
we write its \emph{codomain} ${\sf codom}(f)=\{f(a)\mid a\in A\}$.

The set $\{0,1\}^*$ is called the \emph{infinite binary tree}. 
A \emph{$\Sigma$-tree} (or simply \emph{tree}) is a mapping $\tree:\{0,1\}^*\to\Sigma$. We write $\Trees$ for the set of $\Sigma$-trees. 
The elements of $\{0,1\}^*$ are called \emph{nodes}. 
We call $\epsilon$ the \emph{root} node, and for every node
$\node\in\{0,1\}^*$, $\node\cdot 0$ and $\node\cdot 1$ are called the
\emph{children} of $\node$.
A \emph{branch} in a tree is an infinite sequence of
nodes $\node_0 \node_1 \node_2\ldots$ such that $\node_0=\epsilon$ and for all $i\geq
0$, $\node_{i+1}$ is a child of $\node_i$. Alternatively, a branch $\node_0\node_1\node_2\ldots$ can be seen as an infinite word
$\branch=\lim_{i\rightarrow \infty}\node_i\in\{0,1\}^\omega$.
This way, each node $\node$ induces a cone $\cone(\node)=\node\cdot
\{0,1\}^\omega$ of branches.
Finally, given a branch $\branch\in\{0,1\}^\omega$, we let
$\treelim(\branch)=\tree(\epsilon)\tree(\branch_0)\tree(\branch_0\branch_1)\ldots$
be the sequence of labels along this branch, i.e., we lift $t$ to be a function $t:\{0,1\}^\omega \to \Sigma^\omega$.






Next, we recall various kinds of automata on words and trees
that involve probabilistic aspects: in their transitions and/or their
acceptance conditions.

\subsection{Probabilistic word automata}
\label{sec-word-automata}

A \emph{probabilistic word automaton} $\wauto$ is a tuple
$(Q,\Sigma,\delta,\state_\init,\acc)$ where
\begin{itemize}
\item $Q$ is a finite set of states,
\item $\Sigma$ is an alphabet,
\item $\delta : Q\times \Sigma \times Q\rightarrow [0,1]$ is a probabilistic
transition function, \ie for all $q\in Q$ and $\sigma\in\Sigma$, we
have $\sum_{p\in Q} \delta(q,\sigma,p) = 1$,
\item $\state_\init$ is the initial state, 
\item and $\acc\subseteq Q^\omega$ is an acceptance condition.
\end{itemize}

A \emph{run}   $\run$ of $\wauto$ on $w\in\Sigma^\omega$ is an infinite word over $Q$ such that
$\run_0 = \state_\init$ and for all $i\geq 0$,
$\delta(r_i,w_i,r_{i+1})>0$. A run $\run$ is \emph{accepting} if
$\run\in\acc$. We write $\Runs[\wauto]{w}$ and $\AccRuns[\wauto]{w}$ for the sets of
runs and accepting runs, respectively, of $\wauto$ on $w$.
We recall here certain $\omega$-regular acceptance conditions~\cite{perrin2004infinite}, i.e., 
B\"uchi, co-B\"uchi, Rabin and parity. 
We denote by $\infr$ the set of states that are visited infinitely often along the run $\run$.
The B\"uchi and co-B\"uchi acceptance conditions are given in terms of a set of \emph{accepting} states $\alpha \subseteq Q$.
A run $\run$ is accepting for the B\"uchi acceptance condition iff $\infr \cap \alpha \neq \emptyset$;
and a $\run$ is accepting for the co-B\"uchi acceptance condition iff $\infr \cap \alpha = \emptyset$.
The Rabin acceptance condition is given in terms of Rabin pairs $\{\zug{\alpha_1, \beta_1}, \dots, \zug{\alpha_k, \beta_k}\}$ for some $k \in \setN$ with $\alpha_i, \beta_i \subseteq Q$, and a run $\run$ is accepting if for some $1 \le i \le k$, we have that $\infr \cap \alpha_i \neq \emptyset$ and $\infr \cap \beta_i = \emptyset$.
The parity acceptance condition is defined by a parity function $\alpha: Q \mapsto \{0,1,\dots, k\}$ for some $k \in \setN$.
A run is accepting for the parity condition iff $\displaystyle{\min_{q \in \infr}}\{\alpha(q)\}$ is even, that is, the minimum value seen infinitely often is even.

For every word $w\in\Sigma^\omega$, the automaton induces
a probability distribution $\muw$ on $\Runs{w}$, via cones, $\sigma$-algebras and
 Carath\'eodory's unique extension theorem (see
 \eg\cite{bauer2011measure,carayol2014randomization} for more details).
For convenience we will write $\wauto(w)=\muw(\AccRuns[\wauto]{w})$,
and we call $\wauto(w)$ the \emph{value} of $\wauto$ on $w$.
 While nondeterministic (resp. universal) automata accept a word if
some (resp. every) run is accepting, probabilistic automata allow
for more involved semantics that depend on the value of the automaton
on its input.


\paragraph*{Probable semantics.} 
A word $w$ is \emph{probably accepted} by $\wauto$ if it is accepted
with non-zero probability, \ie $\wauto(w)>0$. The language
$\plang(\wauto)$ is the set of words $w\in\Sigma^\omega$ that
are probably accepted by $\wauto$. The \emph{emptiness problem} for
$\wauto$ with probable semantics asks
whether $\plang(\wauto)=\varnothing$.

\paragraph*{Almost-sure semantics.} 
A word $w$ is \emph{almost-surely
accepted} by $\wauto$ if the set of accepting runs has measure 1, 
that is, $\wauto(w)=1$.
The language $\aslang(\wauto)$ is the set of words
$w\in\Sigma^\omega$ that are almost-surely accepted by $\wauto$.  The
\emph{emptiness problem} for  $\wauto$ with almost-sure semantics asks whether
$\aslang(\wauto)=\varnothing$.


\subsection{Tree automata}
\label{sec-tree-automata}

We first recall non-probabilistic tree automata together with their
classical semantics and the
recent \emph{qualitative
  semantics} of~\cite{carayol2014randomization}.

A \emph{tree automaton} is a tuple $\auto=(\setstates,\Sigma,\trans,\state_\init,\acc)$ where:
\begin{itemize}
\item $\setstates$ is a finite set of states,
\item $\Sigma$ is a finite alphabet,
\item $\trans\subseteq \setstates\times \Sigma\times \setstates\times \setstates$ is a
  transition relation, 
\item $\state_\init\in \setstates$ is an initial state,
\item and $\acc\subseteq Q^\omega$ is an acceptance condition.
\end{itemize}

A \emph{run} of $\auto$ on a $\Sigma$-tree $\tree$ is a $\setstates$-tree $\run$ such that:
\begin{itemize}
\item  $\run(\varepsilon) = \state_\init$
\item $\forall \node\in \words$, we have $(\run(\node),\tree(\node),\run(\node\cdot 0),\run(\node\cdot 1))\in \trans$
\end{itemize}

A branch $\branch\in\{0,1\}^\omega$ of a run $\run$ is \emph{accepting} if
$\treelim[\run](\branch)\in\acc$, and a  run  is accepting if  all its
branches are accepting. 
A run $\run$ is \emph{qualitatively accepting} if 
\[\mu(\{\branch\in\{0,1\}^\omega \mid \treelim[\run](\branch) \in \acc\})=1,\] where
$\mu$ is the coin-flipping probability measure defined on
cones as follows: for $\node\in\{0,1\}^*$, $\mu(\cone(\node))=\frac{1}{2^{|\node|}}$
(see~\cite{bauer2011measure,carayol2014randomization,bojanczyk2016thin}
for more details).

\paragraph*{Tree languages}
We define the \emph{qualitative nondeterministic}
language of a tree automaton $\auto$ as follows:
\begin{align*}
\qelang(\auto)&=\{\tree\pipe \exists \run\, \ist
\run\text{ is a run of }\auto \text{ on }\tree \text{ and }\run\text{ is qualitatively accepting}\}.
\end{align*}

Similarly, we define the \emph{qualitative universal}
language of $\auto$ as follows:
\begin{align*}
  \qulang(\auto)&=\{\tree\pipe \forall \run\, \ist
\run\text{ is a run of }\auto \text{ on }\tree, \run\text{ is qualitatively accepting}\}.
\end{align*}

\subsection{Probabilistic tree automata}


We now recall the probabilistic tree automata introduced
in~\cite{carayol2014randomization}.

A \emph{probabilistic tree automaton} is a tuple $\auto=(\setstates,\Sigma,\delta,\state_\init,\acc)$ where:
\begin{itemize}
\item $\setstates$ is a finite set of states,
\item $\Sigma$ is a finite alphabet,
\item $\delta: \setstates\times \Sigma\times \setstates\times
  \setstates\to [0,1]$ is a probabilistic
  transition function,  \ie for all $q\in Q$ and $\sigma\in\Sigma$, we
have $\sum_{q_0,q_1\in Q} \delta(q,\sigma,q_0,q_1) = 1$,
\item $\state_\init\in \setstates$ is an initial state,
\item and $\acc\subseteq Q^\omega$ is an acceptance condition.
\end{itemize}

A \emph{run} of a probabilistic tree automaton $\auto$ on a $\Sigma$-tree $\tree$ is  a $Q$-tree
$\run$ such that the root is labelled with $q_\init$ and for every
$\node\in\words$, it holds that
$\delta(\run(\node),\tree(\node),\run(\node\cdot 0),\run(\node\cdot 1))>0$. Accepting and
qualitatively accepting runs are defined as before, and the set of runs of
$\auto$  (resp. accepting runs and qualitatively accepting runs) on input tree $t$ is written
$\Runs{t}$ (resp. $\AccRuns{t}$ and $\QualAccRuns{t}$). 
Given a tree $\tree$, one can define a probability measure $\mut$ on
the space of runs (see~\cite{carayol2014randomization}).

\begin{remark}
  \label{rem-measure-runs-trees}
  The definition of runs for probabilistic tree automata in~\cite{carayol2014randomization} allows for
  transitions with probability zero, while we disallow them. But  the
set $R_0$ of all  runs that contain at least one such transition  is a
countable union of cones of partial runs
 of measure zero (this follows directly from the definitions of
 partial runs and cones of runs and their measures, see~\cite[Section
 4.1.1]{carayol2014randomization} for details).
 Therefore $R_0$
has measure zero, and the restriction of the probability measure on
$\Runs{\tree}\cup R_0$ to $\Runs{\tree}$   is a probability measure on $\Runs{\tree}$.
\end{remark}

We define the \emph{almost-sure} and \emph{qualitative almost-sure}
languages of $\auto$ as follows:
\begin{align*}
  \aslang(\auto)&=\{\tree\pipe \mut(\AccRuns{t})=1\}.\\
  \qaslang(\auto)&=\{\tree\pipe \mut(\QualAccRuns{t})=1\}.\\
\end{align*}

As shown in~\cite{carayol2014randomization}, acceptance of trees 
for the qualitative almost-sure semantics can be characterised via
Markov chains, which will be useful later on. 

\begin{definition}
  \label{def-markov-chain}
  A \emph{Markov chain} is a tuple $\mkc=(S,s_\init,\delta,\acc)$
  where
  \begin{itemize}
  \item   $S$
  is a countable set of states,
\item  $s_\init$ is an initial state, 
\item   $\delta:S\times S \to [0,1]$ is a probabilistic transition function
  such that for all $s\in S$, we have $\sum_{s'\in S}\delta(s,s')=1$, and
\item   $\acc\subseteq S^\omega$ is an \emph{objective}.
  \end{itemize}
\end{definition}

A run is an infinite sequence of states, and $\mkc$ induces a
probability measure on runs.
A Markov chain $\mkc$ \emph{almost-surely fulfils} its objective if
the set of runs in $\acc$ has measure one.

\begin{definition}
  \label{def-mk-chain-acc}
  Given a probabilistic tree
  automaton $\auto= (Q,\Sigma, \delta,q_\init,\acc)$  and  a
  $\Sigma$-tree $\tree$, we define the (infinite) Markov chain
  $\mkcAt{\auto}=(S,s_\init,\delta',\acc')$ where:
  \begin{itemize}
  \item $S=Q\times \words \cup Q\times Q \times Q \times \words$,
  \item $s_\init=(q_\init,\epsilon)$,
  \item for all $q,\node,q_0,q_1$,
    \begin{itemize}
    \item $\delta'((q,\node),(q,q_0,q_1,\node))=\delta(q,\tree(\node),q_0,q_1)$,
    \item  $\delta'((q,q_0,q_1,\node),(q_0,\node\cdot
 0)=\delta'((q,q_0,q_1,\node),(q_1,\node\cdot 1)=\frac{1}{2}$, and
    \item  $\delta'(s,s')=0$ in all other cases;
    \end{itemize}
  \item $\acc'$ is inherited from $\acc$: a run is in $\acc'$ if, after
  removing states of the form $(q,q_0,q_1,x)$ and projecting states of
  the form $(q,x)$ on $Q$, we obtain a run in $\acc$. 
\end{itemize}

\end{definition}

The following result is established in~\cite[Proposition 45]{carayol2014randomization}.

\begin{proposition}
  \label{prop-45}
  Let $\auto$ be a probabilistic tree automaton with $\omega$-regular
  acceptance condition, and let $\tree$ be a tree. It holds that
  $\tree\in\qaslang(\auto)$ iff $\mkcAt{\auto}$ almost-surely fulfils its objective.
\end{proposition}


\section{$\qulang$-emptiness is undecidable for parity tree automata}
\label{sec-undec-automata}

In this section we prove our main undecidability result on tree
automata, from which we will derive the undecidability on \MSO in
Section~\ref{sec-MSO}. The undecidability result on tree automata
comes from some undecidability result on word automata. In a few
words, undecidability of the almost-sure emptiness problem was known
to be undecidable for Rabin word automata~\cite{BGB12}. We strengthen this
result to parity word automata with binary branching (for every input letter, each state has exactly
two outgoing transitions with $\frac{1}{2}$ probability). Then, we exploit this result to show undecidability of the
emptiness problem for qualitative parity tree automata. The main lines
 of this proof were 
sketched in an internship report~\cite{PFS14}. The main addition 
we bring is to prove that we can indeed restrict attention to automata
with binary branching, a central assumption  in the proof of the main
result that was not justified in~\cite{PFS14}. 
Further, we do not know whether this assumption can be made in the case of
co-Büchi as claimed in~\cite{PFS14}, but we prove that for Rabin and parity automata one can
indeed assume binary branching while retaining undecidability of the
emptiness problem.
\subsection{Restricting to binary branching}
\label{sec-binary-branching}

We recall the notion of simple automata considered
in~\cite{gimbert2010probabilistic}, and introduce its restriction to
\emph{binary branching}, and a more general class
of \emph{semi-simple} automata, whose emptiness problem  we  prove to
be reducible to the emptiness problem for binary-branching automata.

\begin{definition}
  \label{def-binary-automata}
  A probabilistic word automaton
  $\wauto=(Q,\Sigma,\delta,q_\init,\acc)$ is:
  \begin{itemize}
  \item   \emph{binary
  branching} if ${\sf codom}(\delta)=\{0,\frac{1}{2}\}$;
  \item \emph{simple}
if ${\sf codom}(\delta)=\{0,\frac{1}{2},1\}$;
\item  \emph{semi-simple}
  if ${\sf codom}(\delta) \subseteq \{ \frac{p}{2^q} \mid p,q \in \mathbb{N} \}$. 
\end{itemize}
\end{definition}


In this section we strengthen the following known theorem to binary-branching parity
word automata. It will be used in Section~\ref{sec-words-trees} to establish an undecidability result
for parity tree automata.

\begin{proposition}[\cite{BGB12}]
    \label{prop-as-Rabin}
  The problem whether $\aslang(\wauto)=\varnothing$ is undecidable for
  Rabin word automata. 
\end{proposition}

To strengthen this result to binary-branching parity automata, we need a series of
lemmas. The following result strengthens a result known
from~\cite{BGB12} to simple automata. 

\begin{lemma}
  \label{lem-value-one}
  The problem whether $\plang(\wauto)=\varnothing$ is undecidable for
  simple  B\"uchi word automata. 
\end{lemma}

\begin{proof}
It is proved in~\cite{gimbert2010probabilistic,gimbert:hal-00422888,oualhadj2012value} that the emptiness
problem for simple probabilistic automata on finite words is
undecidable~\cite[Theorem 6.12]{oualhadj2012value}. This result is  used to prove that
the value 1 problem for probabilistic automata on finite words is
undecidable~\cite[Theorem 6.23]{oualhadj2012value}. Since the
reduction in the proof of this result only introduces 
 transitions with probability 1, it holds also for simple automata (see
also~\cite{chatterjee2010probabilistic} for a reformulation of this
construction).

Now it is described in~\cite[Remark 7.3]{BGB12} how to reduce the
value 1 problem for probabilistic automata on finite words to the
emptiness problem for B\"uchi automata with probable semantics. Once
again, this reduction only introduces transitions with probability one,
hence the result.
\end{proof}

Let ${\cal A}$ be a word automaton with a set of accepting states $\alpha$, and we note ${\cal A}_{{\sf B}}$
and ${\cal A}_{{\sf coB}}$ the B\"uchi and
coB\"uchi interpretations of ${\cal A}$, respectively. Then clearly
$\overline{\plang({\cal A}_{{\sf B}})}=\aslang({\cal A}_{{\sf coB}})$.
It is known that probabilistic B\"uchi word automata with probable
semantics are closed under complement~\cite{BGB12}, therefore there
exists a B\"uchi automaton ${\cal A'}_{{\sf B}}$, such that
$\plang({\cal A'}_{{\sf B}})=\aslang({\cal A}_{{\sf coB}})$. While
this implies the undecidability of the emptiness problem for coB\"uchi
word automata with the almost-sure semantics, the automaton
${\cal A'}_{{\sf B}}$ obtained by the complementation procedure
of~\cite{BGB12} is neither simple nor semi-simple in general. To the best
of our knowledge it is open whether the almost-sure emptiness problem
for simple, or even semi-simple, coB\"uchi word automata is decidable
or not although it is claimed to be undecidable in~\cite{PFS14}
without a proof.
%
Here, we prove that the almost-sure emptiness problem for simple Rabin and parity word automata is indeed undecidable.

\begin{lemma}
  \label{lem-buchi-rabin}
  For every simple B\"uchi word automaton $\wauto$ one can construct a
  semi-simple Rabin word automaton $\wauto'$ such that
  $\plang(\wauto)=\plang(\wauto')$ and for every $w\in\Sigma^\omega$,
  $\wauto'(w)\in \{0,1\}$.
\end{lemma}

\begin{proof}
  In~\cite[Theorem 5.3]{BGB12}, it has been proved that for every probabilistic B\"uchi word automaton $\wauto$, there exists a 
  probabilistic Rabin word automaton $\wauto'$ for which for every $w\in\Sigma^\omega$, we have $\wauto'(w)\in \{0,1\}$ and also $\plang(\wauto)=\plang(\wauto')$ holds.
  In the proof of the above theorem, the probabilities of
  the transitions in the Rabin word automaton that is constructed are
finite  sums of finite products of transition probabilities in the original B\"uchi
  automaton, hence the result.
\end{proof}

Note that in Lemma~\ref{lem-buchi-rabin}, $\wauto'$ satisfies that $\aslang(\wauto') =
\plang(\wauto')$. Hence from Lemma~\ref{lem-value-one} and
Lemma~\ref{lem-buchi-rabin} we get:

\begin{lemma}
  \label{lem-rabin-word}
  The problems  whether $\plang(\wauto)=\varnothing$ and whether  $\aslang(\wauto)=\varnothing$ are undecidable for
semi-simple Rabin word automata. 
\end{lemma}





Now, we show how to obtain a simple automaton from a semi-simple automaton while
preserving language emptiness.

\begin{lemma}
  \label{lem-transfo-binary}
  For every semi-simple Rabin word automaton $\wauto$ one can construct a
  simple Rabin word automaton  $\wauto'$ such that  
  $\aslang(\wauto)=\emptyset$ iff $\aslang(\wauto')=\emptyset$, and
  $\plang(\wauto)=\emptyset$ iff $\plang(\wauto')=\emptyset$.
\end{lemma}

\begin{proof}
Let  $\wauto=(Q,\Sigma,\delta,q_\init,\acc)$ be a semi-simple word
automaton, \ie for all $q,q'\in Q$ and $a\in\Sigma$, $\delta(q,a,q')=c/2^d$
for some $c,d\in\setN$.
Since there are finitely many states we can assume that $d$ is the same for all $q,a,q'$ by taking $d$ as the largest of all $d'$ occurring on the transitions and multiplying the constants $c$ accordingly.
For every $q\in Q$ and $a\in\Sigma$, we simulate the possible transitions from
$q$ when reading $a$ with a full binary tree of transitions of depth
$d$, where the root is $q$ and the leaves are the destination states (see Figure~\ref{fig-transfo-simple-binary}, here $d=3$).
To do so we introduce a set of $2^{d}-2$ fresh states $(q,a)_b$ for the
internal nodes of the binary tree of transitions. They are indexed
by all finite words $b\in\{0,1\}^+$ of length at most $d-1$, and the
transitions are as follows: first, they all have probability one half,
except for the last level.
Second, in state $q$ when reading $a$, the two possible transitions
are $(q,a)_0$ and $(q,a)_1$. Then, in all states of the form
$(q,a)_b$, the only transitions with non-zero probability
are by reading the fresh symbol $\#$;
if $b\in\{0,1\}^+$ is of length at most
$d-2$,  it has transitions to $(q,a)_{b\cdot 0}$ and
$(q,a)_{b\cdot 1}$. Finally, for states of the form
$(q,a)_b$ where $b\in\{0,1\}^+$ is of length $d-1$: there are
$2^{d-1}$ such states, and for each one we 
can define two transitions with probability $\frac{1}{2}$, for a total
of $2^d$ possible transitions.
For each $q'\in Q$, if $\delta(q,a,q')=c/2^d$ then we assign $c$
of these possible  transitions to $q'$;  this is possible because
$\sum_{q'\in Q}\delta(q,a,q')=1$.
If a state $(q,a)_b$,
where $b$ is of length $d-1$,
is assigned two outgoing transitions to the same $q'$, we define a transition with probability
$1$
instead. 

Thus $\wauto'=(Q',\Sigma\cup\{\#\},\delta',q_\init,\acc')$ is
defined as follows: $Q'=Q\cup \bigcup_{q,a}Q_{q,a}$, where $Q_{q,a}$
is the set of fresh states of the form $(q,a)_b$. 
The probabilistic transition function
$\delta'$ is defined
as described above. The initial state $q_\init$ is unchanged, and the
acceptance condition $\acc'$ is inherited from $\acc$: a run $\run$ of
$\wauto'$ is in $\acc'$ if its projection $\proj_Q\run$ on $Q$ is in
$\acc$  ($\proj_Q\run$ is obtained by removing from $\run$ states not
in $Q$).
Now one can see that only words of the form $(\Sigma\cdot
\{\#\}^{d-1})^\omega$ have non-zero value in $\wauto'$, and for such a
word $w\in(\Sigma\cdot\{\#\}^{d-1})^\omega$, we have that
$\wauto'(w)=\wauto(\proj_\Sigma(w))$.
As a result there is a bijection between $\aslang(\wauto)$ and
$\aslang(\wauto')$, and also $\plang(\wauto)$ is in bijection with $\plang(\wauto')$.
\end{proof}

  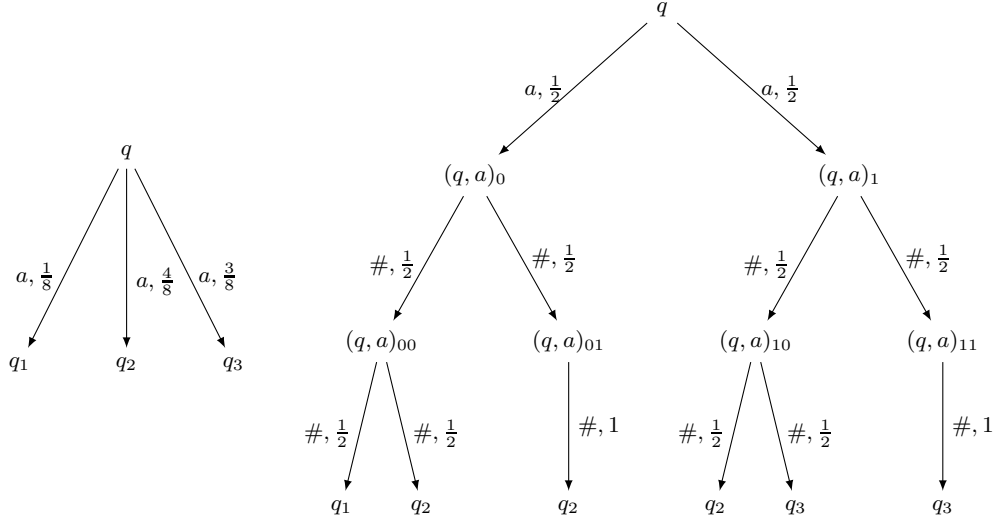
\begin{figure}
      \centering
      \begin{subfigure}{.3\textwidth}
        \centering
         \begin{tikzpicture}
           [
           sibling distance        = 4em,
           level distance          = 8em,
           edge from parent/.style = {draw, -latex},
           every node/.style       = {align=center,font=\footnotesize},
           treenode/.style = {} 
           ]
           \node [treenode] {$q$}
           child { node [treenode] {$q_1$}
             edge from parent node [left,pos=.6] {$a,\frac{1}{8}$}}
           child { node [treenode] {$q_2$}
             edge from parent node [below right] {$a,\frac{4}{8}$}}
           child { node [treenode] {$q_3$}
             edge from parent node [right,pos=.6] {$a,\frac{3}{8}$}};          
         \end{tikzpicture}
       \end{subfigure}
       \begin{subfigure}{.6\textwidth}
         \centering
         \begin{tikzpicture}
           [
           edge from parent/.style = {draw, -latex},
           every node/.style       = {align=center,font=\footnotesize},
           treenode/.style = {}, 
           level distance = 2.2cm,
           level 1/.style = {sibling distance = 14em},
           level 2/.style = {sibling distance = 7em},
           level 3/.style = {sibling distance = 3em}
           ]
           \node [treenode] {$q$}
           child { node [treenode] {$(q,a)_0$}
             child { node [treenode] {$(q,a)_{00}$}
               child { node [treenode] {$q_1$}              
                 edge from parent node [left,pos=.55] {$\#,\frac{1}{2}$}}
               child { node [treenode] {$q_2$}              
                 edge from parent node [right,pos=.55] {$\#,\frac{1}{2}$}}
               edge from parent node [left,pos=.55] {$\#,\frac{1}{2}$}}
             child { node [treenode] {$(q,a)_{01}$}
               child { node [treenode] {$q_2$}              
                 edge from parent node [right] {$\#,1$}}
               edge from parent node [right] {$\#,\frac{1}{2}$}}
             edge from parent node [left] {$a,\frac{1}{2}$}}
           child { node [treenode] {$(q,a)_1$}
             child { node [treenode] {$(q,a)_{10}$}
               child { node [treenode] {$q_2$}              
                 edge from parent node [left,pos=.55] {$\#,\frac{1}{2}$}}
               child { node [treenode] {$q_3$}              
                 edge from parent node [right,pos=.55] {$\#,\frac{1}{2}$}}
               edge from parent node [left,pos=.55] {$\#,\frac{1}{2}$}}
             child { node [treenode] {$(q,a)_{11}$}
               child { node [treenode] {$q_3$}              
                 edge from parent node [right] {$\#,1$}}
               edge from parent node [right] {$\#,\frac{1}{2}$}}
             edge from parent node [right] {$a,\frac{1}{2}$}};          
         \end{tikzpicture}
       \end{subfigure}
         \caption{Transformation from semi-simple to
           simple automata}
    \label{fig-transfo-simple-binary}        
  \end{figure}

Note that for binary-branching automata, for all states
$q\in Q$ and letter $a\in\Sigma$, there are exactly two states
$q_1\neq q_2$ such that $\delta(q,a,q_i)=\frac{1}{2}$, and we may write
$\delta(q,a)=\{q_1,q_2\}$. 
Observe that by duplicating states that
are reached with probability one, every simple probabilistic automaton 
can be easily transformed into an equivalent one with binary
branching. 
We show it for Rabin acceptance condition, but it holds for all $\omega$-regular acceptance conditions.
\begin{lemma}\label{lem-simple-binary}
    For every simple Rabin 
    word automaton $\mathcal{A}$, one can construct 
    a binary-branching Rabin 
    word automaton $\mathcal{B}$ such that 
    $\aslang(\mathcal{A}) = \aslang(\mathcal{B})$ and $\plang(\mathcal{A})=\plang(\mathcal{B})$. 
\end{lemma}
\begin{proof}
Consider a simple word automaton
$\mathcal{A}=(Q,\Sigma,\delta,q_\init,\acc)$ with Rabin (resp. parity)
acceptance condition. 
We construct a binary-branching automaton with Rabin (resp. parity) acceptance condition from $\mathcal{A}$.
First we define $\delta_1 \subseteq \delta$, the set of
transitions that have probability 1: $\delta_1=\{(p,a,q)\in
Q\times\Sigma\times Q\mid
\delta(p,a,q)=1\}$.
We define similarly $\delta_{\frac{1}{2}}$ to be the set of transitions with
probability $\frac{1}{2}$. 
Note that since $\auto$ is simple, for all $(p,a,q)$ that is not in $\delta_1\cup\delta_{\frac{1}{2}}$,
we have that $\delta(p,a,q)=0$.
We also let $Q_1$
be the set of destination states of some transition in $\delta_1$, that is,
$Q_1=\{q \mid \exists p\in Q, \exists a\in \Sigma, (p,a,q)\in\delta_1\}$.
For each state $q \in Q_1$, in the binary-branching automaton, we  create a fresh state $q'$ (the primed version of
$q$) and every transition $(p,a,q)\in\delta_1$ is split into two transitions
$(p,a,q)$ and $(p,a,q')$, each with probability $\frac{1}{2}$.  


Formally, let $Q'_1=\{q'\mid q\in Q_1\}$ be a set of fresh
states.
We construct the binary-branching Rabin (resp. parity) word automaton $\mathcal{B} = (Q',\Sigma,\delta',q_\init,\acc')$, where
$Q'= Q \cup Q'_1$, and
$\delta'$ is defined as follows:
\begin{itemize}
\item for every $(p,a,q)\in\delta_1$ such that $p\notin Q_1$,
  \[\delta'(p,a,q)=\delta'(p,a,q')=\frac{1}{2}\]
\item for every $(p,a,q)\in\delta_1$ such that $p\in Q_1$,
  \[\delta'(p,a,q)=\delta'(p,a,q')=\delta'(p',a,q)=\delta'(p',a,q')=\frac{1}{2}\]
\item for every $(p,a,q)\in\delta_{\frac{1}{2}}$ such that $p\notin
  Q_1$, \[\delta'(p,a,q)=\frac{1}{2}\]
  \item for every $(p,a,q)\in\delta_{\frac{1}{2}}$ such that $p\in Q_1$, \[\delta'(p,a,q)=\delta'(p',a,q)=\frac{1}{2}\]
\end{itemize}
and all other transitions are assigned probability 0 by $\delta'$.


Now we define $\acc'$ for each of $\mathcal{A}$ being a simple Rabin automaton or $\mathcal{A}$ being a simple parity automaton.
First, let $\mathcal{A}$ be a Rabin automaton.
Let $\acc$ be defined in terms of $\{\zug{\alpha_1, \beta_1}, \dots, \zug{\alpha_k, \beta_k}\}$.
We define $\acc'$ in terms of the pairs $\{\zug{\alpha'_1, \beta'_1}, \dots, \zug{\alpha'_k, \beta'_k}\}$, where
$\alpha_i'= \alpha_i \cup \{q'\:|\: q \in \alpha_i$  and $q' \in Q'\setminus Q\}$ and $\beta_i'= \beta_i \cup \{q'\:|\: q \in \beta_i$  and $q' \in Q'\setminus Q\}$ for all $1 \le i \le k$.


From the construction of $\mathcal{B}$, we see that for every word $w \in \Sigma^\omega$, the measure of the set of accepting runs on input $w$ is the same in both $\mathcal{A}$ and $\mathcal{B}$, hence the result.
\end{proof}


Now from Lemma \ref{lem-rabin-word}, Lemma \ref{lem-transfo-binary} and Lemma \ref{lem-simple-binary}, we obtain the following.
 
\begin{corollary}
    \label{thm-as-Rabin}
  The problems whether $\aslang(\wauto)=\varnothing$ and whether $\plang(\wauto)=\varnothing$ are undecidable for
  \emph{binary-branching} Rabin word automata. 
\end{corollary}

Finally,  in the classical (non-probabilistic)
setting, Rabin and parity word automata are known to have the same expressive
power. We show that it also holds under the probabilistic almost-sure
and positive
semantics, while preserving binary branching, and therefore we get the following result:

\begin{theorem}
    \label{thm-as-parity}
    The problems whether $\aslang(\wauto)=\varnothing$ and whether $\plang(\wauto)=\varnothing$ are undecidable for
  \emph{binary-branching} parity word automata. 
\end{theorem}
\begin{proof}
    We show that any binary-branching Rabin word automaton $\wauto$ can be
    converted into an equivalent binary-branching parity word automaton
    $\wauto'$. 

    Let $\wauto = (Q,\Sigma,\delta,q_\init,\acc)$ where $\acc\subseteq Q^\omega$ is a Rabin
    condition (explicitly given as a set of Rabin pairs). We know that
    any (non-probabilistic) Rabin word automaton is effectively
    equivalent to some deterministic parity automaton. Therefore,
    there exists a deterministic parity automaton $P$ over the
    alphabet $Q$ such that its
    language $\lang(P) = \acc$. Let $P = (Q_P,Q, \delta_P, i_P, \alpha)$ where
    $\alpha$ is a 
    parity 
    function. We construct the probabilistic
    parity word automaton $\wauto' = (Q\times Q_P,
     \Sigma,\delta', (q_\init, p), \alpha')$ where 
    \begin{itemize}
      \item $p = \delta_P(i_P, q_\iota)$
      \item $\delta'((q,p), a, (q', p')) = \delta(q,a,q')$ if $p' =
        \delta_P(p,q')$, and $0$ otherwise. 
      \item $\alpha'(q,p) = \alpha(p)$ for all $q\in Q$ and $p\in
        Q_P$. 
    \end{itemize}
    Note that this construction preserves binary branching, and in
    particular we have $\delta'((q,p),a) = \{ (q_1,\delta_P(p,q_1)),
    (q_2,\delta_P(p,q_2))\}$ if $\delta(q,a)=\{q_1,q_2\}$.

    To show that $\aslang(\wauto) = \aslang(\wauto')$ holds, consider
    a word $w\in\Sigma^*$ and an arbitrary linear order $<$ on $Q$. 
    Consider the tree $t_w :
    \{0,1\}^*\rightarrow Q$ defined by $t_w(\epsilon) = q_\iota$ and
    for $u\in\{0,1\}^*$, if
    $\delta(t_w(u),w_{|u|})=\{q_0,q_1\}$ with $q_0<q_1$, then 
    let $t_w(u\cdot i) = q_i$ for $i=0,1$.     We call $t_w$ the
    \emph{tree of runs} on $w$, and let $\acc_w = \{ \branch\in \{0,1\}^\omega\mid
    \treelim[t_w](\branch)\in \acc\}$. 
    The tree $t_w$, with
    probability $\frac{1}{2}$ on all edges, equipped with the acceptance
    condition $\acc_w$, can be seen as an infinite Markov chain which
    almost-surely fulfils its objective iff $w\in \aslang(\wauto)$,
    and fulfils it with positive probability iff $w\in\plang(\wauto)$.

    Similarly, we can define the infinite tree $t'_w :
    \{0,1\}^*\rightarrow Q\times Q_P$ as the tree of runs of $\wauto'$
    on $w$, using any partial order such that $(q_1,p_1) < (q_2,p_2)$ implies $q_1 <
    q_2$. Let also define the acceptance condition $\acc'_w = \{ \branch\in \{0,1\}^\omega\mid
    \treelim[t'_w](\branch)\models \alpha'\}$, which by definition of $\alpha'$ is
    equal to 
    $\{ \branch\in \{0,1\}^\omega\mid \proj_Q(t'_w(\branch))\in
    \acc\}$, where $\proj_Q(t'_w(\branch))$ is the letter-by-letter
    projection of $t'_w(\branch)$ on the $Q$-component. 
    Equipped with $\frac{1}{2}$ probabilities on edges and this
    acceptance condition, $t'_w$ can be seen as an infinite Markov
    chain which almost-surely fulfils its objective iff $w\in
    \aslang(\wauto')$.

    Finally, note that $t_w$ and $t'_w$ are isomorphic, and the
    projection  
    $\proj_Q : Q\times Q_P\rightarrow Q$ 
    allows to
    get $t_w$ from $t'_w$ (by projecting its labels). Moreover, by
    definition of $t'_w$, we also have that $\acc'_w = \acc_w$. Hence,
    seen as infinite Markov chains, $t_w$ and $t'_w$ are the same (up
    to isomorphism). As a consequence, $w\in \aslang(\wauto)$ iff
    $w\in\aslang(\wauto')$, and for the same reason also $w\in \plang(\wauto)$ iff
    $w\in\plang(\wauto')$. 
\end{proof}

\subsection{From words to trees}
\label{sec-words-trees}

In this section we use Theorem~\ref{thm-as-parity} to establish 
an undecidability result for tree automata, but before we recall the
following result which we will use in the proof.

For every probabilistic parity word automaton (PPW) $\wauto = (Q,\Sigma,\delta,q_\init,\acc)$, we define the probabilistic parity tree automaton (PPT) $\auto_\wauto =
(Q,\Sigma,\delta',q_\init,\acc)$ such that for all $p,q\in Q$ and $a\in
\Sigma$,
\begin{itemize}
\item $\delta'(p,a,q,q)=\delta(p,a,q)$, and
\item $\delta'(p,a,q,q')=0$ for $q\neq q'$.
\end{itemize}

\begin{proposition}{\cite{carayol2014randomization}}\label{prop:haddad}
    $\aslang(\wauto)=\varnothing$ iff $\qaslang(\auto_\wauto)=\varnothing$. 
\end{proposition}
\begin{proof}
    In \cite[Proposition 43]{carayol2014randomization}, it is shown that
    $\qaslang(\auto_\wauto)$ is equal to the set of $\Sigma$-trees $t$ such
    that the measure of the branches $\branch$ of $t$ such that $\treelim(\branch)\in
    \aslang(\wauto)$ is $1$.
This immediately yields the result. Indeed, if
    $\aslang(\wauto)=\varnothing$, then no such tree $t$ exists. Conversely,
    if $\aslang(\wauto)$ contains one word $w$, it suffices to construct the
    $\Sigma$-tree $t$ such that for all $\branch\in\{0,1\}^\omega$, we have $\treelim(\branch) =
    w$. Clearly, the measure of the branches $\branch$ of $t$ such that
    $\treelim(\branch)\in \aslang(\wauto)$ is $1$, and therefore $t\in
    \aslang(\auto_\wauto)$.
\end{proof}



We now describe a different construction that translates a binary-branching PPW
$\wauto = (Q,\Sigma,\delta,q_\init,\acc)$ to a PPT $\Aswitch$, which we then show to be
equivalent to $\auto_\wauto$. The PPT 
 $\Aswitch$ is defined as the tuple
 $(Q,\Sigma,\delta',q_\init,\acc)$ where for all states $q,q_1,q_2\in
 Q$ and $a\in
\Sigma$,
\begin{itemize}
\item $\delta'(p,a,q_1,q_2)=\delta(p,a,q_2,q_1)=\frac{1}{2}$, whenever
  $\Delta(q,a) = \{ q_1,q_2\}$,
\item $\delta'(p,a,q_1,q_2)=0$ otherwise.
\end{itemize}

We have the following result:
\begin{lemma}\label{lem-ta}
    Let $\wauto$ be a binary-branching probabilistic parity word
    automaton. Then 
    \[
    \qaslang(\auto_\wauto) = \qaslang(\Aswitch).
    \]
\end{lemma}
\begin{proof}
 The only difference between
$\auto_\wauto$ and        $\Aswitch$   is that 
transitions in $\auto_\wauto$ of the form
        \[
        (q,a,q_1,q_1)\text{ and } (q,a,q_2,q_2)  \text{, each with
          probability }\frac{1}{2},
      \]
become in $\Aswitch$ transitions of the form
        \[
        (q,a,q_1,q_2)\text{ and } (q,a,q_2,q_1) \text{, each with
          probability }\frac{1}{2}.
        \]
      We show that for every tree $\tree$,  the acceptance Markov
      chains $\mkcAt{\auto_\wauto}$ and $\mkcAt{\Aswitch}$ are
      essentially the same. To do so, we construct a Markov chain
      $\mkcAt{}$  that almost-surely fulfils its
      objective iff $\mkcAt{\auto_\wauto}$ 
      almost-surely does, and similarly $\mkcAt{}$ almost-surely fulfils its
      objective iff $\mkcAt{\Aswitch}$ 
      does.
      As a consequence, $\mkcAt{\auto_\wauto}$ almost-surely
      fulfils its objective iff $\mkcAt{\Aswitch}$ does. Hence, by
      Proposition~\ref{prop-45}, we
      get that $\tree\in \qaslang(\Aswitch)$ iff $\tree\in
      \qaslang(\auto_\wauto)$.

      Let us now show how to construct 
      $\mkcAt{}$. We let 
      $$
      \mkcAt{} = (Q\times \{0,1\}^*,
      (q_\iota,\epsilon), \delta_{\mkcAt{}}, \acc_{\mkcAt{}})
      $$ 
      where 
      $\delta_{\mkcAt{}}((q,\node),s)=\frac{1}{4}$ for $s\in\{(q_1,\node\cdot
      0),(q_1,\node\cdot 1),(q_2,\node\cdot 0),(q_2,\node\cdot 1)\}$,
      with $\delta(q,\tree(\node))=\{q_1,q_2\}$, and $\acc_{\mkcAt{}} = \{ \rho\in (Q\times \{0,1\})^*\mid \proj_Q(\rho)\in \acc\}$.

      Observe that $\mkcAt{}$ can be obtained either from $\mkcAt{\auto_\wauto}$
or $\mkcAt{\Aswitch}$ by removing states of
      type $Q^3\times \{0,1\}^*$  and, for each such state,  attaching its children
      to its parent,
      as illustrated in Figure \ref{fig-transfo-mkc}. Indeed, since we have binary
      branching, 
      and by construction of $\Aswitch$, 
      in $\mkcAt{\Aswitch}$ each state of the form $(q,\node)$
      has exactly two successors with $\frac{1}{2}$ probability, of the form
      $(q,q_1,q_2,\node)$ and $(q,q_2,q_1,\node)$. From $(q,q_i,q_j,\node)$, we
      have two $\frac{1}{2}$ probability transitions, one to $(q_i,\node\cdot 0)$
      and one to $(q_j,\node\cdot 1)$. Thus from state $(q,\node)$ we have
      probability $\frac{1}{4}$ to reach each of the states in $\{(q_1,\node\cdot
      0),(q_2,\node\cdot 1),(q_2,\node\cdot 0),(q_1,\node\cdot
      1)\}$, and it is also the case in $\mkcAt{\auto_\wauto}$. Finally, the acceptance condition of
      $\mkcAt{\auto_\wauto}$ and $\mkcAt{\Aswitch}$ are the same as in $\mkcAt{}$, modulo
      projecting paths on states of type $Q\times
      \{0,1\}^*$. Therefore, one gets that $\mkcAt{\Aswitch}$
      almost-surely fulfils its objective iff $\mkcAt{}$ almost-surely
      fulfils its objective iff $\mkcAt{\auto_\wauto}$ almost-surely fulfils its objective.
\end{proof}

We now establish the main result of this section.

\begin{theorem}
  \label{theo-stage}
The problem whether  
$\qulang(\auto)= \varnothing$ is undecidable for parity tree automata. 
\end{theorem}

\begin{proof}
    We reduce the almost-sure emptiness problem of probabilistic parity word
    automata with binary branching, which is undecidable by
    Theorem~\ref{thm-as-parity}. 
    Let $\wauto = (Q,\Sigma,\delta,q_\init,\acc)$ be a probabilistic
    parity word automaton with binary branching. 
    Construct a (non-probabilistic) parity tree automaton $\auto = (Q,\Sigma,\Delta,q_\init,\acc)$ where 
    \[
    \Delta = \{ (q,a,q_1,q_2),(q,a,q_2,q_1)\mid \delta(q,a)=\{q_1,q_2\}\}.
    \]
    We claim that $\aslang(\wauto)=\varnothing$ iff
    $\qulang(\auto)=\varnothing$. 

    \begin{enumerate}

    \item
      $\exists w\in \aslang(\wauto)\implies \exists t\in
      \qulang(\auto)$: Assume that $w\in\aslang(\wauto)$. Construct
      the tree $\tree$ such that for all branches $\branch$, we have
      $t(\branch) = w$. Take any run $\run$ of $\auto$ on $\tree$, and
      define the set 
      $Y=\{\branch\in\{0,1\}^\omega\mid\treelim[\run](\branch)\in\acc\}$
      of accepting branches in $\run$. By definition of
      $\auto$, the run $\run$ (lifted to infinite sequences) is a bijection  
      between $\{0,1\}^\omega$ and
      $\Runs[\wauto]{w}$ that preserves acceptance (i.e., $r(Y) = \AccRuns[\wauto]{w}$), 
      and it also
      induces a bijection
      $f:\node\mapsto \run(\epsilon)\run(\node_0)\ldots\run(\node_0\ldots\node_{|\node|-1})$
      between $\words$ and finite prefixes of runs in
      $\Runs[\wauto]{w}$. 
      We show that $\run$ is measurable, and that $\mu_w$ is the image measure
      of $\mu$ under $\run$, \ie $\mu \circ \run^{-1}=\mu_w$. We
      then conclude that 
      $\mu(Y) = 
      \mu\circ\run^{-1}(\AccRuns[\wauto]{w}) = 
      \mu_w(\AccRuns[\wauto]{w}) = 1$, as required.

      To see that $\run$ is measurable, it is enough to see that for every cone
      $\cone(\rho)\subseteq\Runs[\wauto]{w}$, where $\rho$ is a finite
      prefix of a run in $\Runs[\wauto]{w}$, we have
      $\treelim[\run]^{-1}(\cone(\rho))=\cone(f^{-1}(\rho))$.
      
      We now show that $\mu\circ \run^{-1}$ and $\mu_w$ coincide on
      cones. Then, by Carath\'eodory's unique extension theorem, we
      get that they coincide on all measurable sets.  Let
      $\node\in\words$, and recall that $f$ is a bijection between
      $\{0,1\}^*$ and finite prefixes of runs in
      $\Runs[\wauto]{w}$. On the one hand, because all (non-zero)
      transitions in $\wauto$ have probability $\frac{1}{2}$ and by
      definition of $f$, we have
      $\mu_w(\cone(f(\node)))=\frac{1}{2^{|\node|}}$.  On the other
      hand, by definition of $\run$ and $f$, we have $\mu\circ
      \run^{-1}(\cone(f(\node)))=\mu(\cone(\node))=\frac{1}{2^{|\node|}}$,
      which concludes the proof.
      

      \item $\exists t\in
        \qulang(\auto)\implies \exists w\in \aslang(\wauto)$. Assume
        that $t\in
        \qulang(\auto)$. We
        show that also $t\in \qaslang(\auto_\wauto)$, where
        $\auto_\wauto$ is defined from $\wauto$ as in
        Proposition~\ref{prop:haddad}, from which we get the existence of some $w\in
        \aslang(\wauto)$. 

        Consider automaton $\Aswitch$, defined from $\wauto$ as in
        Lemma~\ref{lem-ta}, and observe that it is a probabilistic version
        of $\auto$ with binary branching. In particular they have same
        states, transitions (except for probabilities), runs, and
        acceptance conditions. 
        Since $t\in \qulang(\auto)$, we also have $t\in
        \qaslang(\Aswitch)$. Indeed, the set of qualitatively accepting runs of
        $\Aswitch$ over $\tree$ is equal to the set of qualitatively accepting
        runs of $\auto$ over $t$. Since
        $\auto$ accepts with a universal condition,
        all runs of $\auto$ over $\tree$ are qualitatively accepting,
        hence the set of qualitatively accepting runs has measure
        $1$. 
        Finally, by Lemma~\ref{lem-ta}, we know that 
        $\qaslang(\Aswitch) = \qaslang(\auto_\wauto)$, hence
        we get that $\tree\in \qaslang(\auto_\wauto)$, concluding the proof.

    \end{enumerate}
  \end{proof}

  Note that a close result was established
  in~\cite{fijalkow2013emptiness}. The difference with ours is that it
  considers alternating co-Büchi automata, while we consider universal
  parity ones. 

 \begin{figure}
     \centering
     \begin{subfigure}{.45\textwidth}
       \centering
        \begin{tikzpicture}
          [
          level distance          = 5em,
          edge from parent/.style = {draw, -latex},
          every node/.style       =
          {align=center,font=\footnotesize,inner sep=1pt},
          treenode/.style = {},
          level 1/.style = {sibling distance = 8em},
          level 2/.style = {sibling distance = 4em}
          ]
          \node at (-1.5,-0.3) {\Large $\mkcAt{\auto_\wauto}$};
          \node [treenode] {$(q,\node)$}
          child { node [treenode] {$(q,q_1,q_1,\node)$}
            child { node [treenode] {$(q_1,\node\cdot 0)$}
              edge from parent node [left=3pt,pos=.55] {$\frac{1}{2}$}}
            child { node [treenode] {$(q_1,\node\cdot 1)$}
              edge from parent node [right=2pt,pos=.55] {$\frac{1}{2}$}}
            edge from parent node [left=3pt,pos=.55] {$\frac{1}{2}$}}
          child { node [treenode] {$(q,q_2,q_2,\node)$}
            child { node [treenode] {$(q_2,\node\cdot 0)$}
              edge from parent node [left=3pt,pos=.55] {$\frac{1}{2}$}}
            child { node [treenode] {$(q_2,\node\cdot 1)$}
              edge from parent node [right=2pt,pos=.55] {$\frac{1}{2}$}}
            edge from parent node [right=3pt,pos=.55] {$\frac{1}{2}$}};          
        \end{tikzpicture}
      \end{subfigure}
           \begin{subfigure}{.45\textwidth}
       \centering
        \begin{tikzpicture}
          [
          level distance          = 5em,
          edge from parent/.style = {draw, -latex},
          every node/.style       =
          {align=center,font=\footnotesize,inner sep=1pt},
          treenode/.style = {},
          level 1/.style = {sibling distance = 8em},
          level 2/.style = {sibling distance = 4em}
          ]
          \node at (-1.5,-0.3) {\Large $\mkcAt{\Aswitch}$};
          \node [treenode] {$(q,\node)$}
          child { node [treenode] {$(q,q_1,q_2,\node)$}
            child { node [treenode] {$(q_1,\node\cdot 0)$}
              edge from parent node [left=3pt,pos=.55] {$\frac{1}{2}$}}
            child { node [treenode] {$(q_2,\node\cdot 1)$}
              edge from parent node [right=2pt,pos=.55] {$\frac{1}{2}$}}
            edge from parent node [left=3pt,pos=.55] {$\frac{1}{2}$}}
          child { node [treenode] {$(q,q_2,q_1,\node)$}
            child { node [treenode] {$(q_2,\node\cdot 0)$}
              edge from parent node [left=3pt,pos=.55] {$\frac{1}{2}$}}
            child { node [treenode] {$(q_1,\node\cdot 1)$}
              edge from parent node [right=2pt,pos=.55] {$\frac{1}{2}$}}
            edge from parent node [right=3pt,pos=.55] {$\frac{1}{2}$}};          
        \end{tikzpicture}
      \end{subfigure}
      \hfill
      
      \begin{subfigure}{.6\textwidth}
        \centering
        \begin{tikzpicture}
          [
          level distance          = 6em,
          edge from parent/.style = {draw, -latex},
          every node/.style       =
          {align=center,font=\footnotesize,inner sep=1pt},
          treenode/.style = {},
          level 1/.style = {sibling distance = 4em},
          level 2/.style = {sibling distance = 4em}
          ]
          \node at (0,1) {};
          \node at (-1,-0.1) {\Large $\mkcAt{}$};
          \node [treenode] {$(q,\node)$}
            child { node [treenode] {$(q_1,\node\cdot 0)$}
              edge from parent node [left=3pt,pos=.55] {$\frac{1}{4}$}}
            child { node [treenode] {$(q_1,\node\cdot 1)$}
              edge from parent node [left=2pt,pos=.55] {$\frac{1}{4}$}}
            child { node [treenode] {$(q_2,\node\cdot 0)$}
              edge from parent node [right=2pt,pos=.55] {$\frac{1}{4}$}}
            child { node [treenode] {$(q_2,\node\cdot 1)$}
              edge from parent node [right=4pt,pos=.55] {$\frac{1}{4}$}};
        \end{tikzpicture}
      \end{subfigure}
        \caption{From  $\mkcAt{\auto_\wauto}$ and $\mkcAt{\Aswitch}$  to $\mkcAt{}$}
   \label{fig-transfo-mkc}        
 \end{figure}

\section{\MSOzeropath on trees}
\label{sec-MSO}

The logic \MSOzero, introduced and studied in
\cite{michalewski2016measure,michalewski2018monadic}, extends
 \MSO with a probabilistic operator $\forall^{=1} X. \varphi$ which
 states that the set of sets satisfying $\varphi$ has Lebesgue-measure
$1$. It is proved in these papers that the \MSOzero-theory of the infinite
binary tree 
is undecidable. They also considered a variant of this
logic, denoted by \MSOzeropath, in which the quantification in the
probabilistic operator is restricted to sets of nodes that form a path. They proved
that, in terms of expressivity, this logic is between  \MSO and \MSOzero, with
a strict gain in expressivity compared to \MSO.
However, they left open the question of the decidability of its theory~\cite[Problem
4]{michalewski2018monadic}.  In this section we establish that it is
in fact undecidable, as a direct consequence of
Theorem~\ref{theo-stage}.

We recall,  from~\cite{michalewski2016measure},  the syntax and semantics of \MSOzeropath on the infinite binary tree.
The syntax of \MSOzeropath is given by the following grammar:
\[\varphi ::= \suc_0(x,y) \pipe \suc_1(x,y) \pipe x\in X 
\pipe \neg \varphi \pipe \varphi_1 \wedge \varphi_2 \pipe \forall x.\varphi \pipe 
\forall X.\varphi \pipe \forallonep X.\varphi\]
where $x$ ranges over a countable set of \emph{first order variables}, and $X$ ranges over a countable set of \emph{monadic second-order variables} (also called \emph{set variables}). The quantifier $\forallonep$ is called the \emph{path-measure quantifier}.

The semantics of \MSO on the infinite binary tree is defined by interpreting the first-order variables $x$ as elements of $\{0,1\}^*$, and the set 
variables $X$ as subsets of $\{0,1\}^*$. Ordinary quantification and the Boolean operations are defined as usual, 
$x \in X$ is interpreted as the membership relation, and $\suc_i$ (for $i = 0,1$) is interpreted as the binary relation 
$\{(x,x\cdot i)\pipe x\in \words\}$. 

We now describe how to interpret the quantification $\forallonep X. \varphi$. A set $X\subseteq \{0,1\}^{*}$ is a $\emph{path}$ if and only if: 
\begin{itemize}
\item $\epsilon\in X$,
\item if $v\in X$ and $w$ is a prefix of $v$ then $w\in X$,
\item if $v\in X$ then either $v\cdot 0\in X$ or $v\cdot 1\in X$, but
  not both.
\end{itemize}

We denote by $\Paths$ the set of all paths. Note that there is a one
to one correspondence between $\Paths$ and the set of branches
$\{0,1\}^\omega$. Thus, the coin-flipping measure $\mu$, defined over
$\{0,1\}^\omega$ (see Section~\ref{sec-tree-automata}), induces a
measure over $\Paths$, which we also denote by $\mu$.  We interpret
$\forallonep X. \varphi$ to mean that the $\mu$-measure of the set of
paths $X$ satisfying $\varphi$ is  $1$.


A sentence is a formula without free variables. The
\emph{\MSOzeropath-theory of the infinite binary tree} is the set of
all \MSOzeropath-sentences $\varphi$ that are true in
the 
infinite binary tree.

Our proof of undecidability will simulate tree automata in the logic. In order to do this, we 
identify sets $X$ with $\{0,1\}$-trees, i.e., the tree associated to $X$ has value $1$ at node $x$ iff $x \in X$. 
In the same way, we identify tuples of variables $X_1,\cdots, X_n$ and $\{0,1\}^n$-trees. This means that an \MSOzeropath formula $\varphi$ with free variables $X_1, \cdots, X_n$ 
can be interpreted on $\{0,1\}^n$-trees.


\begin{theorem}
  The \MSOzeropath-theory of the infinite binary tree is undecidable.
\end{theorem}

\begin{proof}
  The qualitative universal language of a parity tree automaton 
  automaton $\auto$ (over alphabet $\Sigma \subseteq \{0,1\}^n$ for a suitably large $n$) can be expressed in
  \MSOzeropath over the 
  infinite binary tree, i.e., 
  we can construct an \MSOzeropath 
  formula $\varphi_\auto(\vec X)$ such that 
  the set of $\{0,1\}^n$-trees $\vec X = (X_1,\cdots,X_n)$ satisfying $\varphi_\auto$ is equal to 
  $\qulang(\auto)$. The formula $\varphi_\auto(\vec X)$ is of the form 
  \[\forall \vec Y. (\text{``$\vec Y$ is a run of
    $\auto$ on $\vec X$''} \to \forallonep Z.(\text{``$Z$ is an
    accepting path of $\vec Y$''})),\]
where ``$\vec Y$ is a run of
    $\auto$ on $\vec X$'' and ``$Z$ is an
    accepting path of $\vec Y$'' can be expressed in \MSO for parity acceptance conditions (a similar encoding appears in \cite{michalewski2018monadic} for qualitative nondeterministic languages, and in \cite{Rabin:TAMS69} for nondeterministic Muller tree automata). Now, 
    note that $\forall \vec X. \neg \varphi_\auto(\vec X)$ holds in the infinite binary tree if and only if $\qulang(\auto)=\varnothing$. Thus, we have reduced the 
    problem of whether the qualitative universal language of a given parity tree automaton $\auto$ is empty, which is undecidable by Theorem~\ref{theo-stage}, to deciding if the \MSOzeropath sentence $\forall \vec X. \neg \varphi_\auto(\vec X)$ holds in 
    the 
    infinite binary tree. 
    Thus, the \MSOzeropath-theory of the infinite binary tree is undecidable.
\end{proof}




\bibliography{biblio}

\end{document}